\documentclass[journal]{IEEEtran}
\usepackage{amssymb}
\usepackage{amsmath}
\usepackage{graphicx}
\usepackage{cite}

\usepackage{color}
\usepackage{citesort}
\usepackage{multicol}
\usepackage{amsfonts}
\usepackage{bbm}
\usepackage{subfigure}
\usepackage{graphicx,epstopdf}
\usepackage{epsfig}	
\usepackage{cite,graphicx,amsmath,amssymb}
\usepackage{comment}
\usepackage{amsmath}
\usepackage{lipsum}
\usepackage{stfloats}
\usepackage{algorithm}
\usepackage{algorithmic}
\usepackage{makecell}
\newlength{\halfpagewidth}
\divide\halfpagewidth by 2

\newtheorem{theorem}{\textbf{Theorem}}

\newtheorem{lemma}{\textbf{Lemma}}

\newtheorem{corollary}{\textbf{Corollary}}

\newtheorem{proof}{\textbf{Proof}}

\newtheorem{definition}{\textbf{Definition}}

\newtheorem{proposition}{\text{Proposition}}

\makeatletter
\def\ScaleIfNeeded{%
\ifdim\Gin@nat@width>\linewidth \linewidth \else \Gin@nat@width
\fi } \makeatother

\begin{document}
%

\title{Energy and Latency Control for Edge Computing in Dense V2X Networks}

\author{Jingjing Zhao, Lifeng Wang,~\IEEEmembership{Member,~IEEE,}  Kai-Kit~Wong,~\IEEEmembership{Fellow,~IEEE,}\\ Meixia Tao,~\IEEEmembership{Senior Member,~IEEE}, and Toktam Mahmoodi,~\IEEEmembership{Senior Member,~IEEE}
\thanks{The work was supported by EC funded H2020 5GCAR project.}
\thanks{J. Zhao and T. Mahmoodi are with the Department of Informatics, King's College London, London, UK (Email: \{jingjing.2.zhao, toktam.mahmoodi\}@kcl.ac.uk).}
\thanks{L. Wang and K.-K. Wong are with the Department of Electronic and
Electrical Engineering, University College London, London, UK (Email: \{lifeng.wang, kai-kit.wong\}@ucl.ac.uk).}
\thanks{M. Tao is with the Department of Electronic Engineering, Shanghai Jiao Tong University, Shanghai, China (Email: $\rm{mxtao}@sjtu.edu.cn$).}}

\maketitle

\begin{abstract}
This study focuses on edge computing in dense millimeter wave vehicle-to-everything (V2X) networks. A control problem is formulated to minimize the energy consumption under delay constraint resulting from vehicle mobility. { A tractable algorithm is proposed to solve this problem by optimizing the offloaded computing tasks and transmit power of vehicles and road side units. {The proposed dynamic solution can well coordinate the interference without requiring global channel state information}, and  makes a tradeoff between energy consumption and task computing latency.}
\end{abstract}


\vspace{-0.2cm}
\section{Introduction}
Edge computing enables cloud computing capabilities at the edge of wireless networks  for ultra-low latency and high energy efficiency~\cite{Mao_2017}. It is considered as a very prominent technology by cellular and automotive industries and organizations including ETSI and 5GAA for vehicle-to-everything (V2X) networks to support various computation-intensive services for connected vehicles and autonomous driving~\cite{5GAA-VISUAL}. {Existing contributions such as~\cite{Jingyun_2017} have studied autonomous vehicular edge computing via vehicle-to-vehicle (V2V) links between moving vehicles.} Such edge computing services could be used to complete ``cooperative maneuver" or ``cooperative safety" use cases.  For example in the cooperative lane merge scenario, vehicles send their position and speed information to a road side unit
(RSU), based on which the RSU computes an updated local dynamic map. The RSU also computes a list of recommendations for the vehicle, e.g., time to merge and speed of merging. As a consequence, the presence of RSUs can be beneficial not only to host computational power for the data processing, but also to send trajectory recommendation to vehicles.

Task offloading scheduling and resource allocation are critical issues for edge computing in V2X networks due to high speed movement of vehicles. In particular, the maximum number of computing tasks to be offloaded from a vehicle to a  RSU should be controlled carefully to ensure that the offloaded tasks are executed and the computing outputs  are sent back to the vehicle by the same RSU in non-cooperative scenarios. In this work, we study the offloading scheduling and resource allocation for edge computing in dense millimeter wave (mmWave) V2X networks, where there are a large number of connected vehicles
and there is a need for interference coordination among nearby vehicles.  Our aim is to minimize the energy consumption of computing process under delay constraint through optimizing the number of offloaded tasks and transmit powers of vehicles and RSUs.
\vspace{-0.2cm}
\section{{System Model}}

\begin{figure}[t!]
\centering
\includegraphics[width=3.0 in]{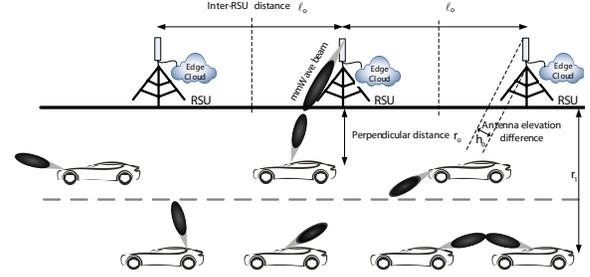}
\caption{An illustration of mmWave V2X network, where there exist vehicle-to-vehicle (V2V) and vehicle-to-RSU (V2I) transmissions.}
\label{RSU_edge}
\vspace{-0.3cm}
\end{figure}
{We consider a  mmWave V2X network, as shown in Fig. \ref{RSU_edge}, where vehicles travel in a highway and RSUs (sites on the lamp posts) are capable of edge computing.} For simplicity, the highway has two lanes\footnote{Note that our model can be easily extended to arbitrary number of lanes.}. It is assumed that {\color{blue}the locations of vehicles in each lane follow independent homogeneous one-dimensional Poisson point process $\Phi_s$ with
the density $\lambda_{s=\{1,2\}}$ ({RSUs can obtain $\lambda_{s=\{1,2\}}$ values via core network from application servers that provide analytics on road traffic, which could be maintained by government agencies or road authorities.}),} and RSUs are evenly distributed with equal distance of $\ell_o$\footnote{Such configuration was also shown in Fig. 3 of~\cite{robert_v2v} where average rate and outage were analyzed in a single RSU downlink transmission case.}.
Each vehicle and RSU are assumed to be equipped with directional antennas with sectored  beam pattern, and the antenna gain for node $i$ in this network is modeled as a function of
steering angle $\phi$ given by,
\begin{align}
\label{array_gain_pattern}
{G_\mathrm{b}^{i}}\left(\phi\right) = \left\{ \begin{array}{l}
G_\mathrm{max}^{i}, \; \mathrm{if} \left|\phi\right|<\phi_\mathrm{b},\\
G_\mathrm{min}^{i}, \; \mathrm{otherwise}
\end{array} \right.
\end{align}
where $\phi_\mathrm{b}$ is the beam-width, $G_\mathrm{max}^{i}$ and $G_\mathrm{min}^{i}$ are main-lobe and side-lobe gains, respectively.
Consider discrete time slots of unit length indexed by $t=1, 2,\dots, T_\mathrm{end}$\footnote{Since the new computing tasks are generated at each time slot and the computing time is reliant on the offloaded tasks at each time slot, the length of each time slot is independent of computing time.}. At each time slot, a vehicle generates computing tasks randomly according to a certain distribution and schedules some tasks to offload them to its associated RSU for execution  based on the task queue state and the channel state.
{\color{blue}{Note that the reasons that a vehicle offloads its computing tasks to the RSU are various, e.g., the hardware incapabilities or lack of big data involving real-time road traffic, road conditions, and parking areas, etc. Such types of computing tasks should be processed at RSUs.}}
Without loss of generality, we consider the computation offloading from a typical vehicle in the first lane, denoted as $o$.
 In the following, we detail
the vehicle's communication channel model, the queue that stores requests/tasks at the vehicle, and the energy consumption model. The optimization problem is formulated to find the optimal transmission  power that allows the vehicle's computing
tasks to be computed within the coverage area of one RSU and keeps the network stable. { RSU acts as the controller, and the control decision is updated at each time slot.}
\vspace{-0.2cm}
\subsection{Dynamic Vehicular Channel Model with Mobility}

Let  $r_o$ denote the  perpendicular distance between the typical vehicle $o$ on the first lane  and each RSU, and $r_1$ denote the  perpendicular distance between a vehicle on the second lane and each RSU, which are assumed to be fixed. Let $h_o$ denote the absolute antenna elevation difference between the vehicle and RSU. Since the effect of small-scale fading on mmWave high-directional communications can be negligible~\cite{TED2013IEEE_Access}, { the three-dimension (3D) line-of-sight (LoS) vehicular channel gain between the typical vehicle $o$ and the associated RSU at time $t$ can be modeled as,
\begin{align}
L_o\left(t\right)=\beta \left({\ell_t^2}+r_o^2+h_o^2\right)^{-\alpha_\mathrm{L}/2},
\end{align}
where {\small{$\ell_t=\left\{ \begin{array}{l}
\frac{{{\ell _o}}}{2} -\bmod \left( {{V_o}t,\frac{{{\ell _o}}}{2}} \right){\rm{,}}\;\; \mathrm{if} \bmod \left( \left\lfloor  \frac{V_ot}{\ell _o/2}\right\rfloor {\rm{,}}2 \right){\rm{ = }}0\\
 \bmod \left( {{V_o}t,\frac{{{\ell _o}}}{2}} \right){\rm{,}}\;\; \mathrm{if} \bmod \left(  \left\lfloor \frac{V_ot}{\ell _o/2} \right\rfloor{\rm{,}}2 \right){\rm{ = }}1
\end{array} \right.$}} is the  horizontal distance  and $\beta$ is the frequency dependent constant value; $\alpha_\mathrm{L}$ is the pathloss exponent, and $V_o$ is
the moving speed of vehicle $o$,} which is assumed to be constant during the time for edge computing service.

\vspace{-0.1 cm}
\subsection{Computing Task Queue Model} 
Let $D_o(t) \in \left\{0,1,\dots, D_{\mathrm{max}}\right\}$ denote the number of new arrival tasks from vehicle $o$ at each time slot $t$, {which is independent of computing time}. The total amount of scheduled tasks for offloading to the RSU at time $t$ is denoted by $C_{\rm in}\left(t\right)$, which needs to be determined at each time slot before computing process.
As such, the computing task queue length $Q_{o}\left(t\right)$ for the vehicle $o$ evolves as follows:
\begin{align}\label{Data_Queue_model}
Q_{o}\left(t+1\right)=\left[Q_{o}\left(t\right)-C_{\rm in}\left(t\right)\right]^{+}+D_o\left(t\right),
\end{align}
where $[x]^+=\max\left\{0,x\right\}$ and $Q_o\left(0\right)=0$.

The total latency of computing process  consists of three parts: 1)  time for uploading tasks from the vehicle to its associated RSU, denoted as $\tau_{1,t}$; 2) time for task execution at the RSU, denoted as $\tau_{2,t}$; and 3) time for downloading computing output from the RSU to the vehicle,  denoted as $\tau_{3,t}$. Considering the fact that the computing output has to be sent back to the vehicle by the same RSU before the vehicle moves to the next RSU, the total number of scheduled tasks $C_{\rm in}\left(t\right)$ in each time slot has to satisfy{\footnote{{Currently, the research on cooperation between RSUs is not addressed.}}}:
\begin{align}\label{offload_constraint}
\underbrace{\frac{C_{\rm in}\left(t\right)}{C_o^{\mathrm{v}}\left(t\right)}}_{\tau_{1,t}}+\underbrace{\frac{\vartheta C_{\rm in}\left(t\right)}{f_\mathrm{R}}}_{\tau_{2,t}}+
\underbrace{\frac{C_{\rm out}\left(t\right)}{C_o^{\mathrm{R}}\left(t\right)}}_{\tau_{3,t}}\leq \frac{\ell_o-\mathrm{mod}(V_o t,\ell_o)}{V_o},
\end{align}
where $\vartheta$ is the number of CPU cycles per bit required for computing  and $f_\mathrm{R}$ is the RSU's CPU clock speed, $C_{\rm out}\left(t\right)$ is the amount of output data after computing (usually $C_{\rm out}\left(t\right) \ll C_{\rm in}\left(t\right)$),  $C_o^{\mathrm{v}}\left(t\right)$ and $C_o^{\mathrm{R}}\left(t\right)$ are the uplink and downlink transmission rate between the vehicle $o$ and its serving RSU, respectively, which are given by
\begin{align}\label{Average_rate}
{\small{C_o^{\mathrm{v}}\left(t\right)=W \log_2\left(1+\frac{P_\mathrm{v}\left(t\right)L_o\left(t\right)G_\mathrm{max}^{\mathrm{v}} G_\mathrm{max}^{\mathrm{R}}}{{\color{blue}I_o^{\mathrm{R}}\left(t\right)}+\sigma^2}\right)}}
\end{align}
with $I_o^{\mathrm{R}}\left(t\right)=\sum\limits_{s{\rm{ = }}1}^2 \sum\limits_{i \in \Phi_s/o} {P_i\left(t\right) {G_\mathrm{b}^{{\rm v}_i}}} G_\mathrm{b}^{{\rm R}} L_i\left(t\right)$, and
\begin{align}\label{Average_rate_1}
\hspace{-0.55cm}{\small{C_o^{\mathrm{R}}\left(t\right)=W \log_2\left(1+\frac{P_\mathrm{R}\left(t\right)L_o\left(t\right) G_\mathrm{max}^{\mathrm{R}}G_\mathrm{max}^{\mathrm{v}} }{\sigma^2}\right)}},
\end{align}
respectively, where $W$ is the mmWave bandwidth, $P_\mathrm{v}\left(t\right)$ and $P_\mathrm{R}\left(t\right)$ are the vehicle $o$ and its serving RSU's transmit power, respectively; $P_i\left(t\right)$ is the interfering vehicle $i$'s transmit power; {$L_i\left(t\right)=\beta {d_i}^{-\alpha_\mathrm{L}}$ is pathloss, in which $d_i$ is the 3D distance between vehicle $i$ and vehicle $o$'s associated RSU; } {\color{blue}$I_o^{\mathrm{R}}\left(t\right)$ is the interference received at the serving RSU,  which  results from  vehicle-to-pedestrian (V2P), vehicle-to-vehicle (V2V) and vehicle-to-RSU (V2I) transmissions}; $\sigma^2$ is the noise power. Note that pessimistic interfering environment is considered, i.e., all the interfering links are LoS. In \eqref{Average_rate_1}, inter-RSU interference is avoided by using narrow beams~\cite{robert_v2v}, and the frequency bandwidths allocated to the vehicle and RSU are orthogonal such that there is no downlink-to-uplink and uplink-to-downlink interference, i.e., cross-link interference is precluded.

\vspace{-0.3 cm}
\subsection{Energy Consumption Model}
We consider computing and communication as the two main contributors to the energy consumption.
The computation energy consumption for the associated RSU at time $t$ is expressed in \eqref{EC1} based on the model in ~\cite{Mao_2017}
\begin{align}\label{EC1}
E_\mathrm{c}^{\mathrm{R}}\left(t\right)=\varrho  C_{\rm in}\left(t\right) \vartheta f_\mathrm{R}^2,
\end{align}
where $\varrho$ is the effective switched capacitance of the RSU processor. Therefore, the whole energy consumption of edge computing for vehicle $o$ at time $t$ is calculated as
\begin{align}
E(t)=E_\mathrm{c}^{\mathrm{R}}\left(t\right)+P_\mathrm{v}\left(t\right)\tau_{1,t}+P_\mathrm{R}\left(t\right)\tau_{3,t}.
\end{align}
\vspace{-0.7cm}

\subsection{Problem Formulation}
While transmitting at the highest power can hypothetically provide the highest rate, and lower the communication latency, it would also increase the interference and inversely impact the reliability as well as latency of the communication. To this end, we formulate an optimization problem to minimize the average energy consumption while keeping the network stable.
 The variables that we optimize are the offloaded computing task $C_{\rm in}\left(t\right)$, transmit powers of the vehicle $o$ and
  its serving RSU $P_\mathrm{v}(t)$ and $P_\mathrm{R}(t)$, at every time slot. Thus, the problem is formulated as
\begin{align}\label{Problem_formulation}
&\mathop {\min }\limits_{ C_{\rm in}, P_\mathrm{v}, P_\mathrm{\mathrm{R}}}   \mathop {\lim }\limits_{T \to \infty } \sup \frac{1}{T}\sum\limits_{t{\rm{ = }}0}^{T - 1} {\mathbb{E}\left[ {E(t)} \right]} \\
&\mathrm{s.t.}~\mathrm{C1}:\mathop {\lim }\limits_{T \to \infty } \sup \frac{1}{T}\sum\limits_{t{\rm{ = }}0}^{T - 1} {\mathbb{E}\left[  {{Q_{o}}\left( t \right)}   \right]}  < \infty , \nonumber\\
&~~~~~\mathrm{C2}: \eqref{offload_constraint}, \forall t,~~~~\mathrm{C3}:  C_{\rm in}\left(t\right) \leq Q_{o}\left(t\right), \forall t,\nonumber\\
&~\mathrm{C4}: C_{\rm in}\left(t\right) \in \mathbb{Z}^+, 0 \leq P_\mathrm{v}\left(t\right) \leq P_\mathrm{v}^{\max} ,   0 \leq P_\mathrm{\mathrm{R}}\left(t\right)\leq
P_\mathrm{R}^{\max},  \forall t, \nonumber
\end{align}
where $P_\mathrm{v}^{\max}$ and $P_\mathrm{R}^{\max}$ are peak power constraints on the vehicle and the RSU, respectively. Constraint C1 ensures that the length of the task queue does not grow unbounded; C2 ensures that offloaded tasks are computed within the required time budget.

\section{Algorithm Design}
{In this section, we develop an online algorithm for solving the stochastic optimization problem \eqref{Problem_formulation} using a  Lyapunov approach~\cite{MJ_Neely_2010}, which only needs the knowledge of the task traffic
  and channel states of the current time slot.} We first define the Lyapunov function as $\mathcal{L}\left(t\right)=\frac{1}{2} Q_o^2\left(t\right)
$,
which is used to measure the task flow of this system.  Then, the expected difference for the Lyapunov function between the time slot $t$ and $(t+1)$ is measured by using the one-time conditional Lyapunov drift, which is
\begin{align}\label{lya_drift}
\Delta\left(t\right)=\mathbb{E}\Big[\mathcal{L}\left(t+1\right)-\mathcal{L}\left(t\right)| Q_o\left(t\right) \Big].
\end{align}
 By introducing the penalty term $\eta \mathbb{E}\Big[E(t)| Q_o\left(t\right) \Big]$ with the non-negative control variable $\eta$, the drift-plus-penalty is
 \begin{align}\label{drift_plus_penalty}
 \Delta_\eta\left(t\right)=\Delta\left(t\right)+\eta \mathbb{E}\Big[E(t)| Q_o\left(t\right) \Big],
 \end{align}
 {\color{blue}where $\eta$ represents the price of energy consumption, lower $\eta$ means that more
 energy will be consumed to accomplish
 more computing tasks. }
\begin{lemma}
For any feasible values of $P_\mathrm{v}\left(t\right)$, $P_\mathrm{R}\left(t\right)$, $\ell_o$, $V_o$, $\eta$ and $Q_o\left(t\right)$, the drift-plus-penalty is upper bounded as
\begin{align}\label{penalty_drift}
 &\Delta_\eta\left(t\right) \leq \mathcal{A}-Q_o\left(t\right)\mathbb{E}\Big[\left(C_{\rm in}\left(t\right)-D_o\left(t\right)\right)| Q_o\left(t\right) \Big] \nonumber\\
& \qquad\qquad~+\eta \mathbb{E}\Big[E(t)| Q_o\left(t\right) \Big],
\end{align}
where $\mathcal{A}$ is a constant satisfying
$\mathcal{A} \geq \frac{\mathbb{E}\Big[{C_{\rm in}^2}\left(t\right)| Q_o\left(t\right) \Big]+D_{\max}^2}{2}$.
\end{lemma}

 The proof of \textbf{Lemma 1} can be easily obtained by following the approach in~\cite[Chapter 4]{MJ_Neely_2010}.
{In light of stochastic network optimization, control decision is made at every time $t$ for minimizing the upper bound of drift-plus-penalty given in the right-hand-side (RHS) of \eqref{penalty_drift}, which guarantees that constraint C1 of problem \eqref{Problem_formulation} is met~\cite{MJ_Neely_2010}. Thus, the task offloading and transmit power control can be made at each time slot by solving the following problem:
\begin{align}\label{Problem_11_final}
\hspace{-0.2cm}\mathop {\min }\limits_{{C_{\rm in}\left(t\right),P_\mathrm{v}\left(t\right), P_\mathrm{\mathrm{R}}\left(t\right) }} &  - Q_o\left(t\right)C_{\rm in}\left(t\right)+\eta E_\mathrm{c}^{\mathrm{R}}\left(t\right) \nonumber\\
&\hspace{-0.2cm}~+\eta\left(P_\mathrm{v}\left(t\right)\frac{C_{\rm in}\left(t\right)}{C_o^{\mathrm{v}}\left(t\right)}+P_\mathrm{R}\left(t\right)\frac{C_{\rm out}\left(t\right)}{C_o^{\mathrm{R}}\left(t\right)}\right)\\
\mathrm{s.t.}\qquad &\mathrm{C2}, \mathrm{C3}, \mathrm{C4}.\nonumber
\end{align}}
 {The above is a mixed integer and non-convex problem, due to the integral of $C_{\rm in}\left(t\right)$ and non-convexity of the objective function and constraint $\mathrm{C2}$. Moreover, the amount of interference received at the typical vehicle's associated RSU is difficult to evaluate in such fast changing and complicated networks. To make the problem \eqref{Problem_11_final} more tractable, we relax $C_{\rm in}\left(t\right) \in \mathbb{Z}^{+}$ to $C_{\rm in}\left(t\right)\geq 0$,}  and transform problem \eqref{Problem_11_final} as,
\begin{align}\label{Problem_Trans1}
&\mathop {\min }\limits_{{C_{\rm in}\left(t\right),P_\mathrm{v}\left(t\right), P_\mathrm{\mathrm{R}}\left(t\right) }}  - Q_o\left(t\right)C_{\rm in}\left(t\right)+\eta E_\mathrm{c}^{\mathrm{R}}\left(t\right) \nonumber\\
&\quad+\eta\left(P_\mathrm{v}\left(t\right)\frac{C_{\rm in}\left(t\right)}{\overline{C}_o^{\mathrm{v}}\left(t\right)}+P_\mathrm{R}\left(t\right)
\frac{C_{\rm out}}
{{C}_o^{\mathrm{R}}\left(t\right)}\right)\\
& \mathrm{s.t.} ~~ \mathrm{C2}: \frac{C_{\rm in}\left(t\right)}{\overline{C}_o^{\mathrm{v}}\left(t\right)}+\frac{\vartheta C_{\rm in}\left(t\right)}{f_\mathrm{R}}+
\frac{C_{\rm out}}{{C}_o^{\mathrm{R}}\left(t\right)}\leq \frac{\ell_o-\mathrm{mod}(V_o t,\ell_o)}{V_o},  \nonumber\\
&\qquad \mathrm{C3},~~\mathrm{C4},~~\mathrm{C5}: I_o^{\mathrm{R}}\left(t\right) \leq  I_\mathrm{th}^{\mathrm{R}},\nonumber
\end{align}
where $\overline{C}_o^{\mathrm{v}}\left(t\right)=W \log_2\bigg(1+\frac{P_\mathrm{v}\left(t\right)L_o\left(t\right)G_\mathrm{max}^{\mathrm{v}} G_\mathrm{max}^{\mathrm{R}}}{I_\mathrm{th}^{\mathrm{R}}+\sigma^2}\bigg)$, $I_\mathrm{th}^{\mathrm{R}}$ is the maximum uplink  interference temperature. Thus, constraint C5 represents the maximum uplink interference that can be tolerated. We note that the transformed problem \eqref{Problem_Trans1} provides a performance lower bound of the original problem \eqref{Problem_11_final}, since the lower-bounded rate $\overline{C}_o^{\mathrm{v}}\left(t\right)$ is adopted and less tasks will be offloaded for satisfying C2. However, in practice, it is still challenging to solve \eqref{Problem_Trans1} because of the unpredictable interference $I_o^{\mathrm{R}}\left(t\right)$
at time $t$ in C5. {\color{blue}The following Lemma 2  provides an alternative to satisfy C5 without requiring global channel
state information such as the interfering vehicles' locations and moving speeds, etc.}
\begin{lemma}\label{lemma_1}
{Given an arbitrary small $\epsilon$, $\Pr\left(I_o^{\mathrm{R}}\left(t\right) \geq  I_\mathrm{th}^{\mathrm{R}} \right) \leq \epsilon$  when the vehicle's transmit power satisfies
$P_\mathrm{v}\left(t\right)\leq \frac{\epsilon I_\mathrm{th}^{\mathrm{R}}}{\Xi_1 \Upsilon}$, where $\Xi_1=\mathbb{ E}\left[{G_\mathrm{b}^{{\rm v}_i}} G_\mathrm{b}^{{\rm R}} \right]$ is the uplink average antenna gain from an interfering vehicle to the RSU, and $\Upsilon$ is given by}
\begin{align}\label{upsilon_11}
\Upsilon=&2{\lambda _1} \beta \int_{\frac{{{\ell _o}}}{2}}^\infty  \left(x^2+r_o^2+h_o^2\right)^{-\alpha_\mathrm{L}/2}     dx \nonumber\\
&\qquad+2{\lambda_2} \beta \int_{\frac{{{\ell _o}}}{2}}^\infty  \left(x^2+r_1^2+h_o^2\right)^{-\alpha_\mathrm{L}/2}     dx.
\end{align}
\end{lemma}
\begin{proof}
According to the Markov's inequality, we have
\begin{align}\label{lema_1}
\Pr\left(I_o^{\mathrm{R}}\left(t\right) \geq  I_\mathrm{th}^{\mathrm{R}} \right) \leq \frac{\mathbb{E}\left[I_o^{\mathrm{R}}\left(t\right) \right]}{I_\mathrm{th}^{\mathrm{R}}}.
\end{align}
By using the Campbell's theorem~\cite{Baccelli2009}, $\mathbb{E}\left[I_o^{\mathrm{R}}\left(t\right) \right]$ is derived as
\begin{align}\label{exp_sum_interf}
&\mathbb{E}\left[I_o^{\mathrm{R}}\left(t\right) \right]=\mathbb{E}\left[\sum\limits_{i \in \Phi_1/o} {P_\mathrm{v}\left(t\right) {G_\mathrm{b}^{{\rm v}_i}}} G_\mathrm{b}^{{\rm R}} L_i\left(t\right)\right]+ \nonumber\\
&\qquad \mathbb{E}\left[\sum\limits_{i \in \Phi_2} {P_\mathrm{v}\left(t\right) {G_\mathrm{b}^{{\rm v}_i}}} G_\mathrm{b}^{{\rm R}} L_i\left(t\right)\right]= P_\mathrm{v}\left(t\right)\Xi_1 \Upsilon.
\end{align}
Note that in \eqref{exp_sum_interf}, {the minimum horizontal distance between the interfering vehicles and the RSU associated to the typical vehicle should be larger than $\frac{\ell_o}{2}$}, i.e., {\color{blue}the intra-RSU interference resulting from V2X transmissions can be mitigated by using multiple access techniques and directional beamforming.} Let $\Pr\left(I_o^{\mathrm{R}}\left(t\right) \geq  I_\mathrm{th}^{\mathrm{R}} \right) \leq \frac{\mathbb{E}\left[I_o^{\mathrm{R}}\left(t\right) \right]}{I_\mathrm{th}^{\mathrm{R}}} \leq \epsilon$ with ${\mathbb{E}\left[I_o^{\mathrm{R}}\left(t\right) \right]}$ given by \eqref{exp_sum_interf}, we have $P_\mathrm{v}\left(t\right)\leq \frac{\epsilon I_\mathrm{th}^{\mathrm{R}}}{\Xi_1 \Upsilon}$.
\end{proof}

\textbf{Lemma 2} shows that the level of vehicle's transmit power depends on the density of vehicles, in order to control interference. Moreover, {\color{blue}by leveraging narrower beams, the amount of interference is reduced and $\Xi_1$ becomes lower, which  allows larger transmit power to improve transmission rate and thus reduces latency.}  Since the transmit power allocation is carried out at each time slot, the time slot index $t$ is omitted in the following analysis  for notation simplicity. By relaxing the constraint C5  with the help of \textbf{Lemma 2}, problem \eqref{Problem_Trans1} is rewritten as
\begin{align}\label{probem_yy_p22}
&\mathop {\min }\limits_{{{C_{\rm in}},P_\mathrm{v}, P_\mathrm{\mathrm{R}} }}  - Q_o {C_{\rm in}}+\eta E_\mathrm{c}^{\mathrm{R}} +\eta\left(P_\mathrm{v}\frac{{C_{\rm in}}}{\overline{C}_o^{\mathrm{v}}}+P_\mathrm{R}
\frac{C_{\rm out}}
{{C}_o^{\mathrm{R}}}\right)\\
& \mathrm{s.t.} ~~ \mathrm{C2},~~\mathrm{C3},~~\mathrm{C4},~~\mathrm{C5}:  P_\mathrm{v}\leq  \frac{\epsilon I_\mathrm{th}^{\mathrm{R}}}{\Xi_1 \Upsilon}. \nonumber
\end{align}
Compared to problem \eqref{Problem_Trans1}, the solution of problem \eqref{probem_yy_p22} is robust and does not require global channel state information for interference coordination. To solve the above problem, we adopt the decomposition approach, considering the fact that ${C_{\rm in}}$ and $P_\mathrm{v}$ are coupled in the objective function of \eqref{probem_yy_p22}. { As such, problem \eqref{probem_yy_p22} is decomposed into two subproblems, i.e., for fixed  ${C_{\rm in}}$, we have
\begin{align}\label{probem_yy_p22_1}
&\mathop {\min }\limits_{{P_\mathrm{v}, P_\mathrm{\mathrm{R}} }} U_1=P_\mathrm{v}\frac{{C_{\rm in}}}{\overline{C}_o^{\mathrm{v}}}+P_\mathrm{R}
\frac{C_{\rm out}}
{{C}_o^{\mathrm{R}}}\\
& \mathrm{s.t.} ~~ \mathrm{C2}, \mathrm{C4}, \mathrm{C5}, \nonumber
\end{align}
and for fixed $\left(P_\mathrm{v},P_\mathrm{R}\right)$, we have the following linear program
\begin{align}\label{1D_LP}
&\mathop {\min }\limits_{{{C_{\rm in}}} }  - Q_o {C_{\rm in}}+\eta \varrho {C_{\rm in}} \vartheta f_\mathrm{R}^2 +\eta P_\mathrm{v}\frac{{C_{\rm in}}}{\overline{C}_o^{\mathrm{v}}}
\\
& \mathrm{s.t.} ~~ \mathrm{C2},~~\mathrm{C3}. \nonumber
\end{align}}

Since problem \eqref{probem_yy_p22_1} is still non-convex, we propose a successive convex approximation (SCA)-based method. Let ${\bf{y}}=\left(P_\mathrm{v}, P_\mathrm{\mathrm{R}}\right)$ denote the variable vector to optimize   and  ${\bf{y}}\left(\upsilon\right)=\left(P_\mathrm{v}\left(\upsilon\right), P_\mathrm{\mathrm{R}}\left(\upsilon\right)\right)$ denote the solution at the $\upsilon$-th iteration.  According to the  idea of SCA, the solution is obtained by
successively solving a sequence of subproblems as:
\begin{align}\label{semi_conduct11_11}
&{\hat{\bf{y}}}\left({\bf{y}}\left(\upsilon\right)\right)=\mathop {\mathrm{argmin} }\limits_{\bf{y}} \widetilde{U_1} \left(\bf{y};{\bf{y}}\left(\upsilon\right)\right) \\
& \mathrm{s.t.} ~~\mathrm{C2}, \mathrm{C4}, \mathrm{C5}, \nonumber 
\end{align}
where $\widetilde{U_1} \left(\bf{y};{\bf{y}}\left(\upsilon\right)\right)$ is the approximation for the objective function of problem \eqref{probem_yy_p22_1} at the $\upsilon$-th iteration,
 given by
\begin{align}\label{transform_tilde_U}
&\widetilde{U_1} \left(\bf{y};{\bf{y}}\left(\upsilon\right)\right)= P_\mathrm{v}\left(\upsilon\right)\frac{{C_{\rm in}}}{\overline{C}_o^{\mathrm{v}}}+ P_\mathrm{v}\frac{{C_{\rm in}}}{\overline{C}_o^{\mathrm{v}}\left(\upsilon\right)}+P_\mathrm{R}\left(\upsilon\right)
\frac{C_{\rm out}}{{C}_o^{\mathrm{R}}}\nonumber\\
&+P_\mathrm{R}
\frac{C_{\rm out}}{{C}_o^{\mathrm{R}}\left(\upsilon\right)}+\left(\bf{y}-{\bf{y}}\left(\upsilon\right)\right)^{T} {\bf{\Phi}}\left(\bf{y}\left(\upsilon \right)\right) \left(\bf{y}-{\bf{y}}\left(\upsilon\right)\right),
\end{align}
where ${\bf{\Phi}}\left(\bf{y}\left(\upsilon\right)\right)$ is a diagonal matrix with any positive entries and could depend
on $\bf{y}\left(\upsilon\right)$.
 The subproblem \eqref{semi_conduct11_11} in each iteration is strongly convex and can be solved by utilizing the Lagrangian-duality method~\cite{Convex_Book}.
 \textbf{Algorithm 1} outlines the main steps to solve  problem \eqref{probem_yy_p22_1}. In the procedure, the measure of optimality  $M({\bf{y}}\left(\upsilon \right))$ is defined based on~\cite{Cannelli_2017}, and the iteration complexity is $\mathcal{O}\left(\zeta^{-1}\right)$.
\setlength{\tabcolsep}{3 pt} \begin{table}[htbp] 
\centering
\begin{tabular}{l}
\hline
{\textbf{Algorithm 1} SCA-based Algorithm }\\ \hline
1: \quad Initialize feasible values ${\bf{y}}\left(0\right)$, $\Phi\left(\bf{y}\right)$. \\
    \qquad \quad  Let $\alpha=10^{-5}$, $\zeta=10^{-5}$, and $\upsilon=0$. \\
2: \quad  Calculate ${\hat{\bf{y}}}\left({\bf{y}}\left(\upsilon\right)\right)$ from  \eqref{semi_conduct11_11} by using Lagrangian-duality method. \\
3: \quad $\bf{if}$ $\parallel M({\bf{y}}\left(\upsilon \right)) \parallel_{2}^2 \leq \zeta$, stop. \\
4: \quad  Set ${\bf{y}}\left(\upsilon+1 \right)={\bf{y}}\left(\upsilon \right)+\delta\left(\upsilon\right)\left({\hat{\bf{y}}}\left({\bf{y}}\left(\upsilon\right)\right)-{\bf{y}}\left(\upsilon \right)\right)$ \\~~with the diminishing step-size rule $\delta\left(\upsilon\right)=\delta\left(\upsilon-1\right)\left(1-\alpha \delta\left(\upsilon-1\right)\right)$. \\
5: \quad Set $\upsilon\leftarrow\upsilon+1$, and return step 2. \\
\hline
\end{tabular}
\label{table:1}
\end{table}

After obtaining the solution $\left(P^{*}_\mathrm{v},P^{*}_\mathrm{R}\right)$ of problem \eqref{probem_yy_p22_1}, the corresponding ${C_{\rm in}}$ is
updated by solving problem \eqref{1D_LP}, which is given in three different cases:
\begin{itemize}
\item Case 1: $ \eta=0$. In this case, the energy consumption is ignored, and vehicle $o$ only needs to maximize the amount of computing tasks at each time slot. The optimal $C_\mathrm{in}$ is  ${C^*_{\rm in}}=\min\left\{\frac{\frac{\ell_o-\mathrm{mod}(V_o t,\ell_o)}{V_o}-\frac{C_{\rm out}}{{C}_o^{\mathrm{R}}}}{\frac{1}{\overline{C}_o^{\mathrm{v}}}+\frac{\vartheta }{f_\mathrm{R}}},Q_o\right\}$ with $P^*_\mathrm{v}=\min\left\{ \frac{\epsilon I_\mathrm{th}^{\mathrm{R}}}{\Xi_1 \Upsilon},P_\mathrm{v}^{\max}\right\}$ and $ P^*_\mathrm{R} = P_\mathrm{R}^{\max}$ in light of the constraints $\mathrm{C2}-\mathrm{C5}$ given in \eqref{probem_yy_p22}.

  \item Case 2: $0< \eta<\frac{Q_o}{\varrho \vartheta f_\mathrm{R}^2 +\frac{P^{*}_\mathrm{v}}{\overline{C}_o^{\mathrm{v}}}}$. In this case,  the priority is to offload the maximum allowable amount of computing tasks for a typical vehicle at this time slot. {Thus, ${C^*_{\rm in}}=\min\left\{\frac{\frac{\ell_o-\mathrm{mod}(V_o t,\ell_o)}{V_o}-\frac{C_{\rm out}}{{C}_o^{\mathrm{R}}}}{\frac{1}{\overline{C}_o^{\mathrm{v}}}+\frac{\vartheta }{f_\mathrm{R}}},Q_o\right\}$ with the solution $\left(P^{*}_\mathrm{v},P^{*}_\mathrm{R}\right)$ of problem \eqref{probem_yy_p22_1}.}

  \item Case 3 :  $\eta \geq \frac{Q_o}{\varrho \vartheta f_\mathrm{R}^2 +\frac{P^{*}_\mathrm{v}}{\overline{C}_o^{\mathrm{v}}}}$. In this case,  ${C^*_{\rm in}}=0$ for energy saving,  since energy consumption needs to be controlled at this time slot.
\end{itemize}

As such, the solution of problem \eqref{probem_yy_p22} can be iteratively obtained, which is a robust and suboptimal solution of original problem \eqref{Problem_11_final}.  Note that since we have relaxed the variable $C_{\rm in}$, the desired $C^*_{\rm in}$ is the nearest allowable amount of computing tasks to the obtained solution ${C_{\rm in}}$ of problem \eqref{probem_yy_p22},  to ensure $C_{\rm in} \in \mathbb{Z}^{+}$.
As such, we propose \textbf{Algorithm 2} to solve our stochastic optimization problem \eqref{Problem_formulation}.

\setlength{\tabcolsep}{3 pt}\begin{table}[htbp]  \label{Stochastic_optimization} 
\centering
\begin{tabular}{l}
\hline
\textbf{Algorithm 2} Dynamic Computing-aware Power Allocation  Algorithm\\ \hline
1: {\bf{if}} $t=0$, \bf{then} \\
2: \quad{\bf{Initialize}} the control variable $\eta=\eta_o$, ${C_{\rm in}}\left(0\right)=D_o\left(0\right)$. \\
    \qquad  Observe the computing task queue length ${Q_{o}}\left( t\right)$.\\
3: \bf{else}  \\
4: \quad {\bf{repeat}}\\
5:  \qquad  Computing task and transmit power control decisions:   \\
\qquad \qquad  \textbf{Loop:}\\
 \qquad \qquad~~ a) Given $C_{\rm in}\left(t\right)$, \\
 \qquad \qquad~~~~~~~~~update $P_z\left(t\right), z \in \left\{\mathrm{v}, \mathrm{R} \right\} $ via \textbf{Algorithm 1}. \\
\qquad \qquad~~ b) Update $C_{\rm in}\left(t\right)$ based on the solution of problem \eqref{1D_LP}.\\
\qquad \qquad  \textbf{Until convergence}. \\
6: \qquad $t=t+1$.   \\
7: \qquad Update the computing task queue length based on \eqref{Data_Queue_model}.  \\
8: \quad \textbf{Until} $t=T_{\mathrm{end}}$. \\
9:  \bf{end if}\\
\hline
\end{tabular}
\label{table 1}
\end{table}

\section{Numerical Results}
\setlength{\tabcolsep}{0.25 pt}\begin{table*}[thb]   
\centering
\caption{Simulation Parameters}\label{table1}
\begin{tabular}{|c|c|c|c|}
 \hline
Parameter & {Value} & {Parameter} & {Value}\\
 \hline
\makecell{ Absolute antenna elevation \\difference between vehicle and RSU } &  6 m & \makecell{ Perpendicular distance between \\a vehicle on the 1st lane  and RSU} & 7 m\\
 \hline
  \makecell{Perpendicular distance between \\a vehicle on the 2nd lane  and RSU }& 10 m &   RSU site-distance & 50 m\\
  \hline
  Vehicle speed & 50 km/h & Pathloss exponent & 2.0\\
   \hline
Carrier frequency   & $f_c=60$ GHz &  Frequency dependent constant value & $\left(\frac{{{3 \times 10^8}}}{{4\pi {f_c}}}\right)^2$\\
  \hline
 System bandwidth   & 2 GHz &  RSU maximum transmit power & 35 dBm \\
  \hline
  Vehicle maximum transmit power & 25 dBm & Noise power &  $-174+10 \times \log_{10}(\mathrm{Bandwidth})+7 (\mathrm{Noise\;figure})$ dBm\\
  \hline
RSU main-lobe gain &  15 dB & RSU side-lobe gain &  -15 dB \\
 \hline
Vehicle main-lobe gain &  3 dB & Vehicle side-lobe gain &  -3 dB \\
 \hline
 RSU beamwidth & $9^{o}$ & Vehicle beamwidth & $90^{o}$ \\
 \hline
 Each task size & $10*10^6$ (bits) & RSU computing capacity & $10*10^{9}$ CPU cycles/s \\
 \hline
 Required CPU cycles per bit & 300 & Switched capacitance constant & $10^{-28}$ \\
 \hline
\end{tabular}
\end{table*}
{\color{blue}We simulate the formulated problem, and run it for $T_{\mathrm{end}}=3000$ time slots.} The remainder of parameters are summarized in Table I.  The number of new tasks per time slot follows a Poisson distribution with the density $\lambda_\mathrm{Task}$, and the output data size (in bits) follows a uniform distribution over the interval $\left[1,10^6\right]$. {\color{blue}Since the amount of interference cannot be easily obtained in such a complicated and fast-changing network with mobility, we consider the worst interfering environment that the amount of interference in each time slot is equal to the level of  interference temperature, and thus calculate the energy consumption per time slot in a lower-bound form as}
{\begin{align*}
E_\mathrm{Low}(t)=E_\mathrm{c}^{\mathrm{R}}\left(t\right)+P^{*}_\mathrm{v}\left(t\right)\frac{C^{*}_{\rm in}\left(t\right)}{\overline{C}_o^{\mathrm{v}}\left(t\right)}+P^{*}_\mathrm{R}\left(t\right)\frac{C_{\rm out}}{{C}_o^{\mathrm{R}}\left(t\right)},
\end{align*}
where $\left(P^{*}_\mathrm{v}\left(t\right),C^{*}_{\rm in}\left(t\right),P^{*}_\mathrm{R}\left(t\right)\right)$ is the optimal control solution at time $t$ of problem \eqref{Problem_11_final} }. Note that we also provide the maximum level of energy consumption for comparison. Such a case occurs in the noise-limited scenario (i.e., vehicles are not very dense.).

Fig.~\ref{fig1} shows the effects of different levels of interference temperature and densities of vehicles on the delay and energy consumption. {\color{blue}The average delayed number of tasks for execution and the average energy consumption are calculated as $\sum\limits_{t{\rm{ = }}1}^{{T_{end}}} {{Q_o}\left( t \right)}/T_{\mathrm{end}}$ and $\sum\limits_{t{\rm{ = }}1}^{{T_{end}}} {E_\mathrm{Low}\left( t \right)}/T_{\mathrm{end}}$, respectively.} It is obvious that there exists a tradeoff between the task execution delay and energy consumption, since more
tasks that can be computed by RSU will inevitably bring more energy cost. When the vehicle's density is lower, more tasks can be offloaded to the RSU for reducing vehicle's computing burden since larger vehicle and RSU's transmit power is allowed. {\color{blue}When vehicles become denser, lower interference temperature means that vehicle's transmit power will be
controlled in a much lower level, which results in higher delay during uplink  transmissions, and thus less tasks can be
offloaded to RSU.} It is confirmed that {\color{blue}by using the proposed algorithm, interference can be well coordinated, and the
achievable performance in a dense case can be close to that of noise-limited case (where there are few interfering vehicles in the lane, e.g., $\lambda_1=\lambda_2=0.001/\mathrm{m}$ in Fig. 2.).}

\begin{figure}[htbp]
\vspace{-0.7 cm}
\centering
\includegraphics[width=3.6 in,height=2.4 in]{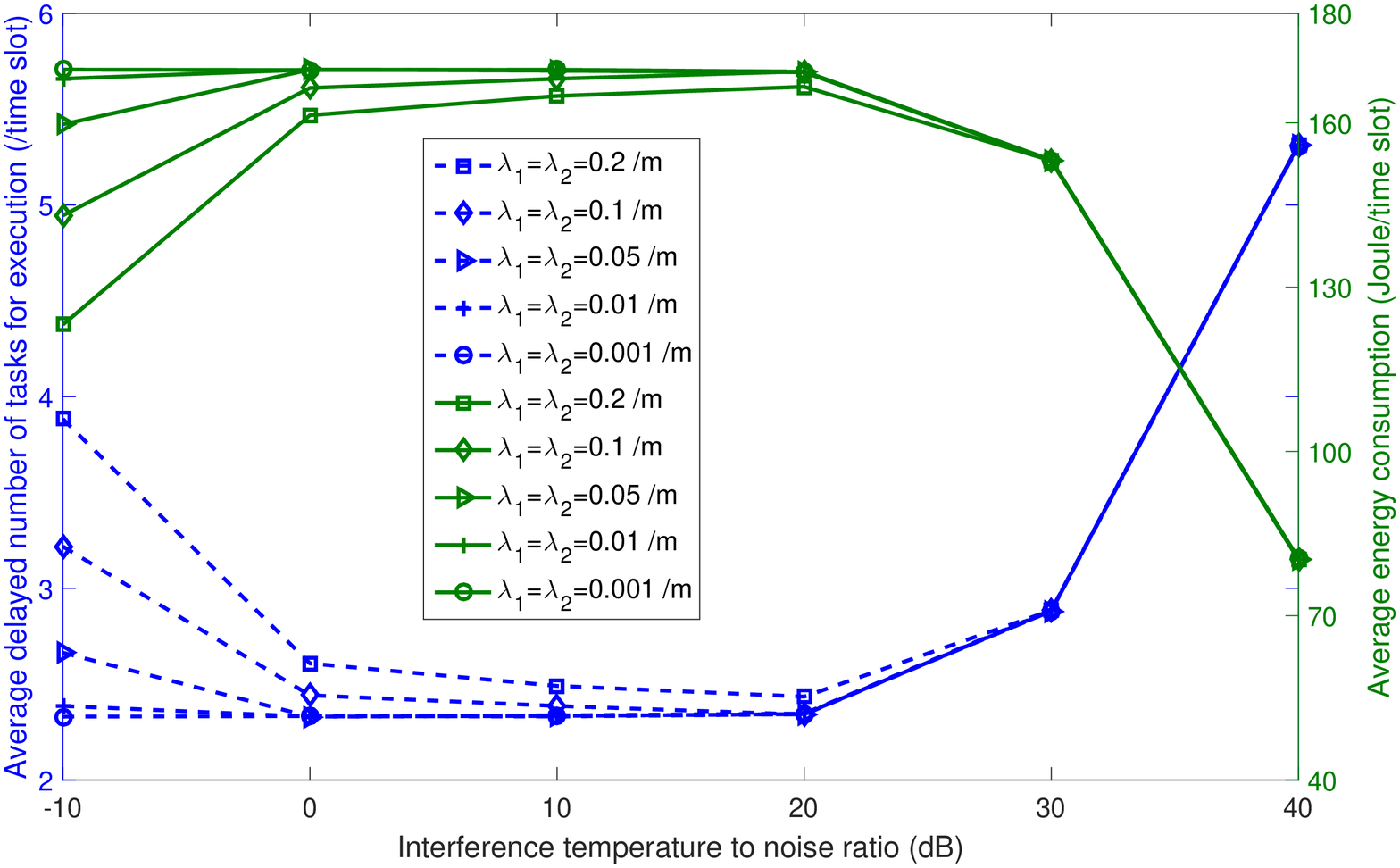} 
\caption{Effects of different levels of interference temperature and densities of vehicles: $\lambda_\mathrm{Task}=8/\mathrm{time~slot}$, $\epsilon=0.1$ and $\eta=10^{14}$.} \vspace{-0.3cm}
\label{fig1}
\end{figure}

%

\begin{figure}[htbp]
\vspace{-0.1 cm}
\centering
\includegraphics[width=3.6 in,height=2.4 in]{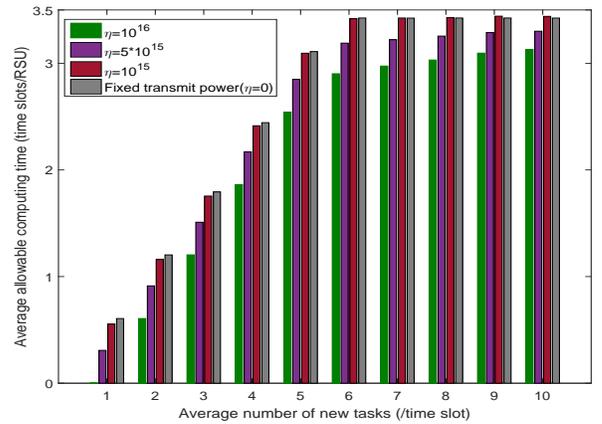}
 \caption{Effects of different average numbers of tasks per time slot $\lambda_\mathrm{Task}$ and control values $\eta$: $I_\mathrm{th}^{\mathrm{R}}=20\;\mathrm{dB}$, $\lambda_1=\lambda_2=0.1/\mathrm{m}$ and $\epsilon=0.1$.}
\label{fig2}
\vspace{-0.4 cm}
\end{figure}

Fig.~\ref{fig2} shows the effects of different $\lambda_\mathrm{Task}$ and $\eta$ on the average allowable computing time in an RSU range.
 {\color{blue} The average allowable computing time is calculated as $\sum\limits_{t{\rm{ = }}1}^{{T_{\mathrm{end}}}} {\sum\limits_{j{\rm{ = }}1}^3 {{\tau _{j{\rm{,}}t}}} }/{T_{\mathrm{end}}}$ where ${\tau _{j{\rm{,}}t}}$ is defined in \eqref{offload_constraint}.} It is obvious that more computing  tasks consumes more time. The average allowable computing time converges to a constant value for larger $\lambda_\mathrm{Task}$, since computing time is limited to meet the constraint given by \eqref{offload_constraint}. {{\color{blue}Compared to the case that energy consumption is ignored (i.e., the control variable $\eta=0$), the proposed solution can well control the computing time by selecting proper $\eta$ value, to achieve a tradeoff between energy consumption and computing latency.}}

\section{Conclusion}
This paper considered edge computing empowered dense mmWave V2X networks, in which stochastic geometry is adopted to model such a random network with vast vehicle connections.  By using Lyapunov and SCA-based optimization theory, an efficient online algorithm was proposed. The results confirmed that the proposed algorithm can efficiently control the amount of energy consumption and computing tasks.




\bibliographystyle{IEEEtran}

\end{document}